\documentclass[oribibl]{llncs}
\usepackage{amsmath, amssymb, algorithm, algorithmic, color}
\usepackage{enumerate, graphicx, hyperref}
\usepackage{cite}

\let\doendproof\endproof
\renewcommand\endproof{~\hfill\qed\doendproof}

\usepackage[margin=1in]{geometry}

\title{Grid Minors in Damaged Grids}
\author{David Eppstein}
\institute{Department of Computer Science, University of California, Irvine}

\begin{document}
\maketitle

\begin{abstract}
We prove upper and lower bounds on the size of the largest square grid graph that is a subgraph, minor, or shallow minor of a graph in the form of a larger square grid from which a specified number of vertices have been deleted. Our bounds are tight to within constant factors. We also provide less-tight bounds on analogous problems for higher-dimensional grids.
\end{abstract}

\pagestyle{plain}

\section{Introduction}
Let $\mathcal{F}$ be a fixed minor-closed family of graphs,\footnote{Recall that a \emph{minor} of an undirected graph $G$ is a graph formed from $G$ by edge contractions, edge deletions, and vertex deletions, and a class $\mathcal{F}$ of graphs is \emph{minor-closed} if every minor of a graph in $\mathcal{F}$ remains in $\mathcal{F}$. To avoid exceptional cases we do not consider the family of all graphs to be minor-closed. The \emph{treewidth} of a graph $G$ is one less than the maximum clique size of a graph that contains $G$ as a subgraph, has no induced cycles of length greater than three, and minimizes the clique size among graphs with these properties.}  let $G$ be a graph in $\mathcal{F}$ with treewidth $t$, and suppose that $k$ vertices are deleted from $G$. What can we say about the treewidth of the remaining graph $G'$? If we did not constrain $G$ to belong to $\mathcal{F}$, the obvious answer would be that the treewidth of $G'$ is at most $t-k$, as the $(t+1)$-clique provides an example for which this bound is tight. However, in this paper we show that membership of $G$ in $\mathcal{F}$ causes $G'$ to have treewidth that, for sufficiently large values of $k$, is significantly larger than $t-k$. For instance, when $k=t$, the $t-k$ bound becomes trivial, but the remaining treewidth will still always be proportional to $t$ (with a constant of proportionality depending on $\mathcal{F}$). Our proof of this result follows from a reformulation of the treewidth problem as a seemingly much more specialized question about grid graphs, where we define an $n\times n$ grid graph to be the graph with $n^2$ vertices at the points $(i,j), 0\le i,j<n$ of the two-dimensional integer lattice, with edges connecting points at unit distance apart. If the $n\times n$ grid graph is damaged by the removal of $m$ of its vertices, how large a grid must exist as a minor of the remaining graph?

Questions about the size of grid minors are central to Robertson and Seymour's work on the theory of graph minors, and notoriously difficult. Following earlier analogous results for infinite graphs~\cite{Hal-MN-65,Die-AMSUH-04}, Robertson and Seymour showed the existence of a non-constant function $f$ such that every graph of treewidth~$t$ has a grid minor of size $f(t)\times f(t)$, and subsequent researchers have provided improved bounds on~$f$~\cite{RobSey-JCTB-91,RobSeyTho-JCTB-94,Ree-SC-97,DieJenGor-JCTB-99,BirBonRee-DAM-09,CheChu-13}. The growth rate of $f$ is not known: its known upper and lower bounds are both polynomial in~$t$, but with different exponents~\cite{CheChu-13}. There has also been much research on related problems for special classes of graphs~\cite{RobSeyTho-JCTB-94,DemHaj-Comb-08,DemHajKaw-Algo-09,GuTam-Algo-12,KawKob-STACS-12,Gri-DMTCS-11} or other structures than grids~\cite{ReeWoo-EJC-12}. In particular, for every fixed minor-closed family $\mathcal{F}$, and for all graphs $G$ in $\mathcal{F}$, the treewidth of $G$ and the side length of the largest grid minor of $G$ are within constant factors of each other (with the factors depending on~$\mathcal{F}$)~\cite{DemHaj-Comb-08}. It follows that the questions above about the effect on treewidth of vertex deletions in a minor-closed family $\mathcal{F}$ and about the largest grid minor in a damaged grid have answers that are within a constant factor of each other.

Finding smaller but undamaged grids within damaged grids is also a classical question from distributed computing, where the grid graph is assumed to represent the processors and communication links of an ideal distributed computing system, and one would like to be able to run the system as if it were in its ideal state even when some of its processors have become faulty. Cole, Maggs, and Sitaraman~\cite{ColMagSit-SJC-97} showed that a faulty system can simulate the ideal system as long as the number of faults is at most $n^{1-\epsilon}$ for some $\epsilon>0$. However, their simulation does not find a grid-like subgraph of the remaining undamaged part of the initial grid. Instead, it involves loading some processors and links with the work that in the original grid would have gone to a multiple processors and links, and it has $\Theta(n)$ startup time making it efficient only for simulations with high running time. Indeed, prior research had already shown that, if an ideal grid is to be emulated by a damaged grid of the same size by directly embedding the ideal grid into the damaged grid (with constant load, dilation, and congestion) then only a constant number of faults can be tolerated~\cite{KakKarLei-FOCS-90}. A graph-theoretic formulation of this distributed computing problem  that is intermediate between the subgraph and simulation approaches is to ask for the largest \emph{shallow grid minor} in a damaged grid. Here the shallowness of the minor encapsulates the requirement that each communications link in the minor be simulated by a small number of links in the original network~\cite{PloRaoSmi-SODA-94,Ten-CGTA-98,Wul-FOCS-11,NesOss-12}.

A third motivation for the damaged grid problem comes from recent work on approximation algorithms for the minimum genus of a graph embedding. Chekuri and Sidiropoulos~\cite{CheSid-ms-12} show that it is possible to find an embedding whose genus is at most a fixed polynomial of the genus of the optimal embedding; their method uses the problem of finding a large grid in a damaged grid as one of its steps. An early version of their work observed that an $n\times n$ grid graph with $m$ damaged vertices always has an undamaged grid \emph{subgraph} of size $O(\tfrac{n}{\sqrt{m}})\times O(\tfrac{n}{\sqrt{m}})$. This bound turns out to be tight to within a constant factor. However, their method can use minors in place of subgraphs, and using the bounds on grid minors that we prove here allowed them to reduce the exponent of the polynomial describing their solution quality.

In this paper we determine to within a constant factor the size of the largest grid minor that is guaranteed to exist in a damaged grid: it is $\Theta(\min\{n,n^2/m\})$. In particular, if $cn$ vertices are deleted from an $n\times n$ grid (for any constant $c>0$), the remaining graph still necessarily contains a grid minor of linear size.
We then generalize this result in two ways, to higher dimensional grids (both grids of bounded dimension and unbounded side length, and hypercubes of unbounded dimension) and to shallow minors. Our results are constructive, in the sense that they lead directly to simple and practical algorithms for finding grid minors in damaged grids, avoiding the high constant factors of many results in graph minor theory.

As a notational convention, we use $\log$ to mean the base two logarithm, and $\ln$ to mean the natural logarithm. In many cases these functions appear inside $O$-notation in which case the constant factor distinguishing them is irrelevant.

\section{Planar grid subgraphs}
For completeness, we briefly repeat the argument of Chekuri and Sidiropoulos showing that every $n\times n$ grid with $m$ damaged vertices has an undamaged grid of size $O(\tfrac{n}{\sqrt{m}})\times O(\tfrac{n}{\sqrt{m}})$, and that this is tight.

\begin{theorem}[Chekuri and Sidiropoulos]
\label{thm:subgraph}
If $G$ is an $n\times n$ grid graph, and $D$ is a set of $m$ vertices in $G$, then $G$ has a grid of size $k\times k$ that is disjoint from $D$, where $k=\left\lfloor\frac{n}{\lceil\sqrt{m+1}\rceil}\right\rfloor$.
\end{theorem}

\begin{proof}
Partition $G$ into $\lceil\sqrt{m+1}\rceil\times\lceil\sqrt{m+1}\rceil>m$ smaller grids, of size $k\times k$ (possibly with some rows and columns left over). Because there are more than $m$ of these smaller grids, at least one has to be disjoint from $D$.
\end{proof}

\begin{figure}[t]
\centering\includegraphics[height=2in]{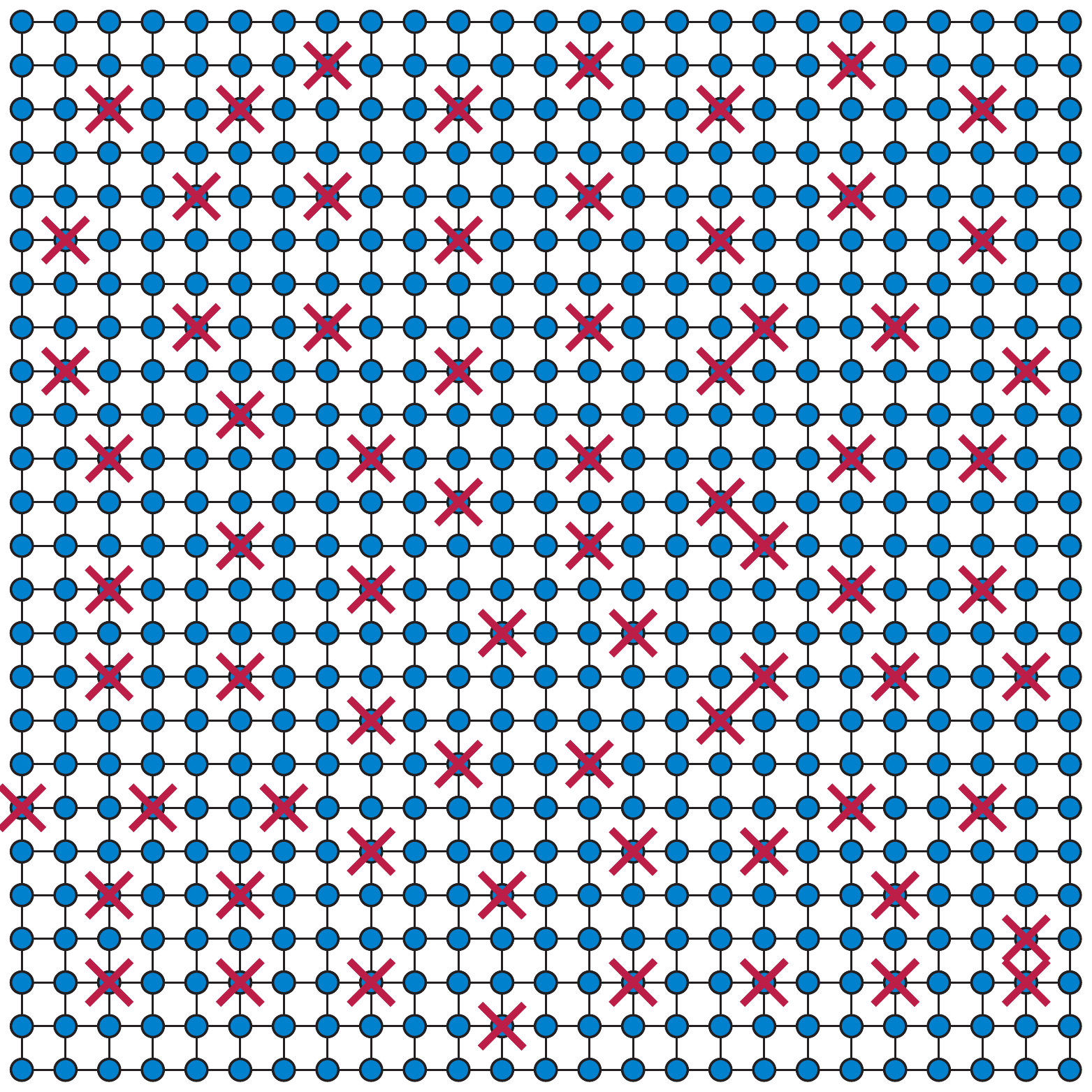}\qquad
\includegraphics[height=2in]{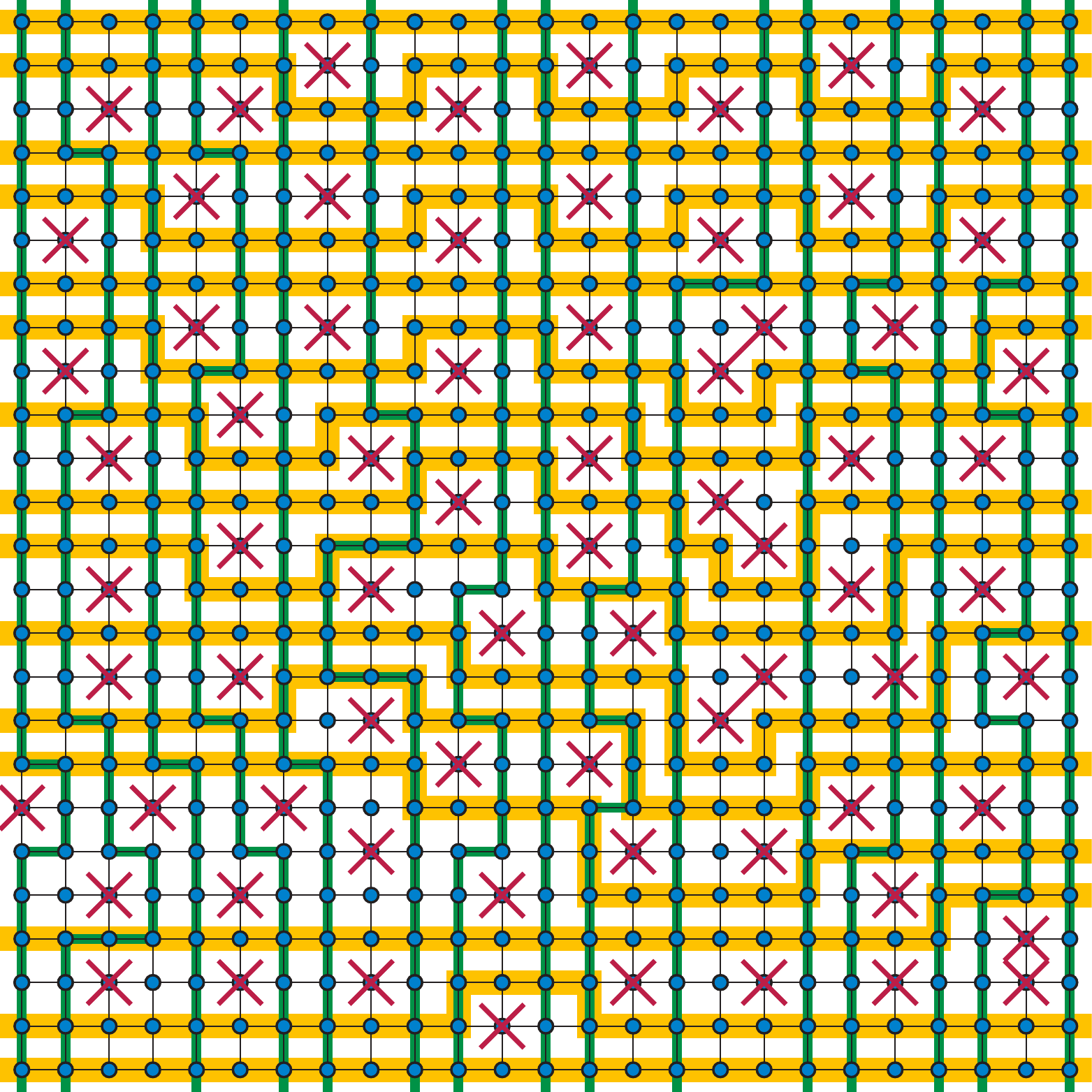}
\caption{Left: A $25\times 25$ grid with 72 damaged vertices. There are no undamaged $3\times 3$ grid subgraphs, but many undamaged $2\times 2$ subgraphs. Right: Sets of disjoint paths from one side of the same damaged grid to the other, avoiding the damaged vertices and forming a $15\times 15$ grid minor.}
\label{fig:grid-holes}
\end{figure}

For instance, Figure~\ref{fig:grid-holes}(left) has $n=25$ and $m=72$.  For these numbers, $\lceil\sqrt{m+1}\rceil=9$ and $k=2$. For this example, the bound of Theorem~\ref{thm:subgraph} is tight: there are no $3\times 3$ undamaged subgrids, but if we partition it into 144 small $2\times 2$ subgrids (with one row and column left over) then many of them must be undamaged.

A matching upper bound on the size of a grid subgraph is also possible.
If we place the vertices of $D$ at positions whose coordinates are both congruent to $-1$ modulo $k+1$, in a coordinate system for which one of the grid corners is $(0,0)$, then the total number of vertices placed is $\lfloor n/(k+1)\rfloor^2$, and the largest remaining square grid subgraph has size $k\times k$.
The inequality $\lfloor n/(k+1)\rfloor^2\le m$ has as its maximal solution the same choice of $k$ as in Theorem~\ref{thm:subgraph}, $k=\left\lfloor\frac{n}{\lceil\sqrt{m+1}\rceil}\right\rfloor$.

Our bound for grid minors will use this same basic idea as Theorem~\ref{thm:subgraph}, of partitioning into smaller grids in order to make the damage in at least one subgrid sparser than in the original grid, but will find a subgrid that is only lightly damaged rather than one that is not damaged at all.

\section{Planar grid minors}

When forming a grid minor, rather than a grid subgraph, we may tolerate a greater amount of damage to the original grid, because there are more ways of forming minors than subgraphs.
One very versatile way of forming grid minors is by finding many disjoint paths across the grid, as shown in Figure~\ref{fig:grid-holes}(right).

\begin{lemma}
\label{lem:disjoint-paths}
Suppose that an $n\times n$ grid is damaged by the deletion of a set of $m$ vertices, but that we can find a collection of $k$ vertex-disjoint paths extending horizontally from one vertical side of the grid to the other, and another collection of $k$ vertex-disjoint paths extending vertically from one horizontal side of the grid to the other. Suppose also that each horizontal-vertical pair has a connected intersection. Then the damaged grid contains a $k\times k$ grid minor.
\end{lemma}

\begin{proof}
Delete any edges and vertices of the grid that do not belong to the paths.
Contract each intersection between a horizontal and vertical path to form each of the grid vertices, and contract the portion of each path between two of these intersections to form the grid edges.
A version of the Jordan curve theorem ensures that the intersections on each path lie in the correct order, so the contracted graph is the desired grid.
\end{proof}

Figure~\ref{fig:grid-holes}(right) shows the paths of the lemma as yellow and green. In this example, there are $k=15$ paths of each type, so the lemma gives us a $15\times 15$ grid minor, much larger than the largest undamaged grid subgraph.

\begin{figure}[t]
\centering\includegraphics[height=2in]{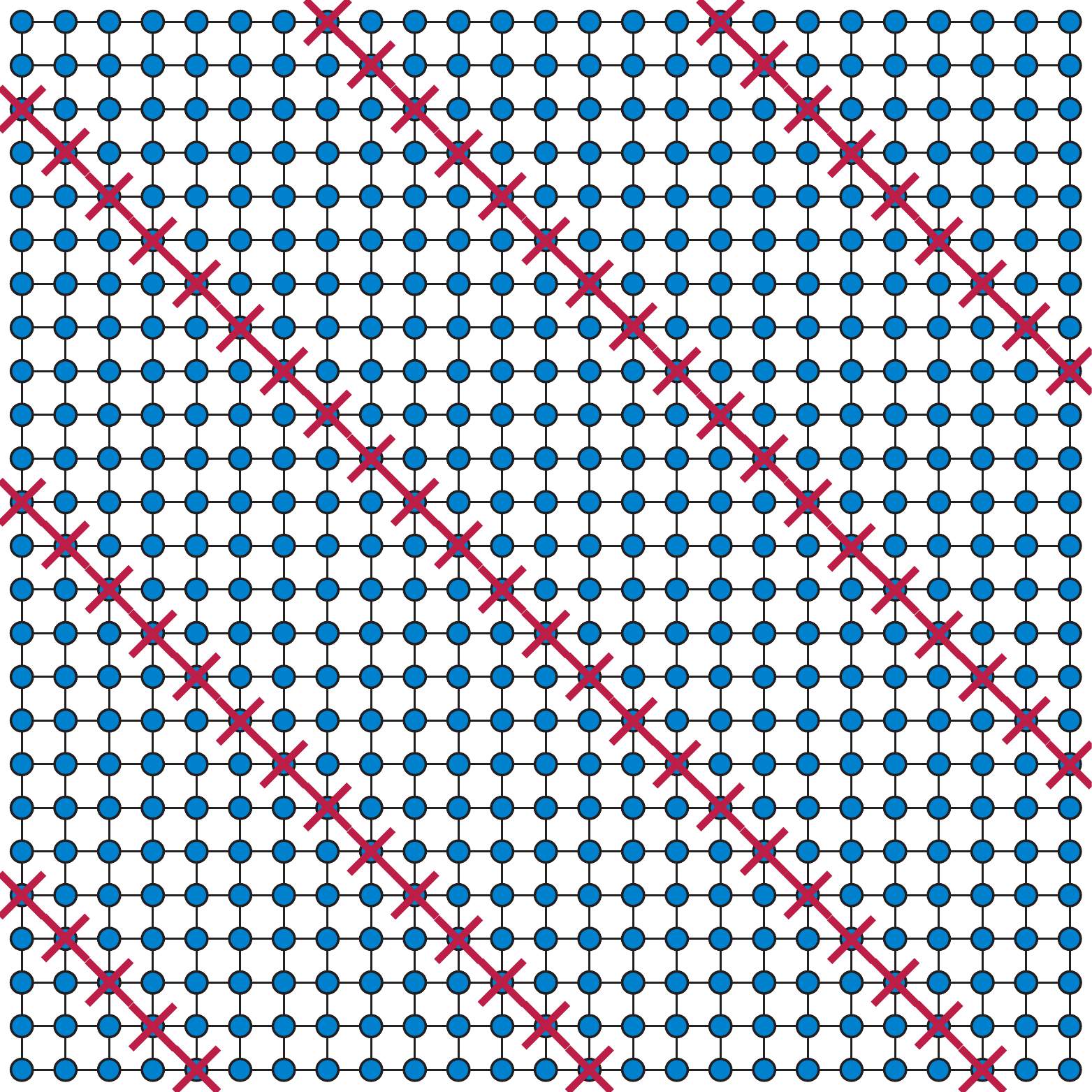}\qquad
\includegraphics[height=2in]{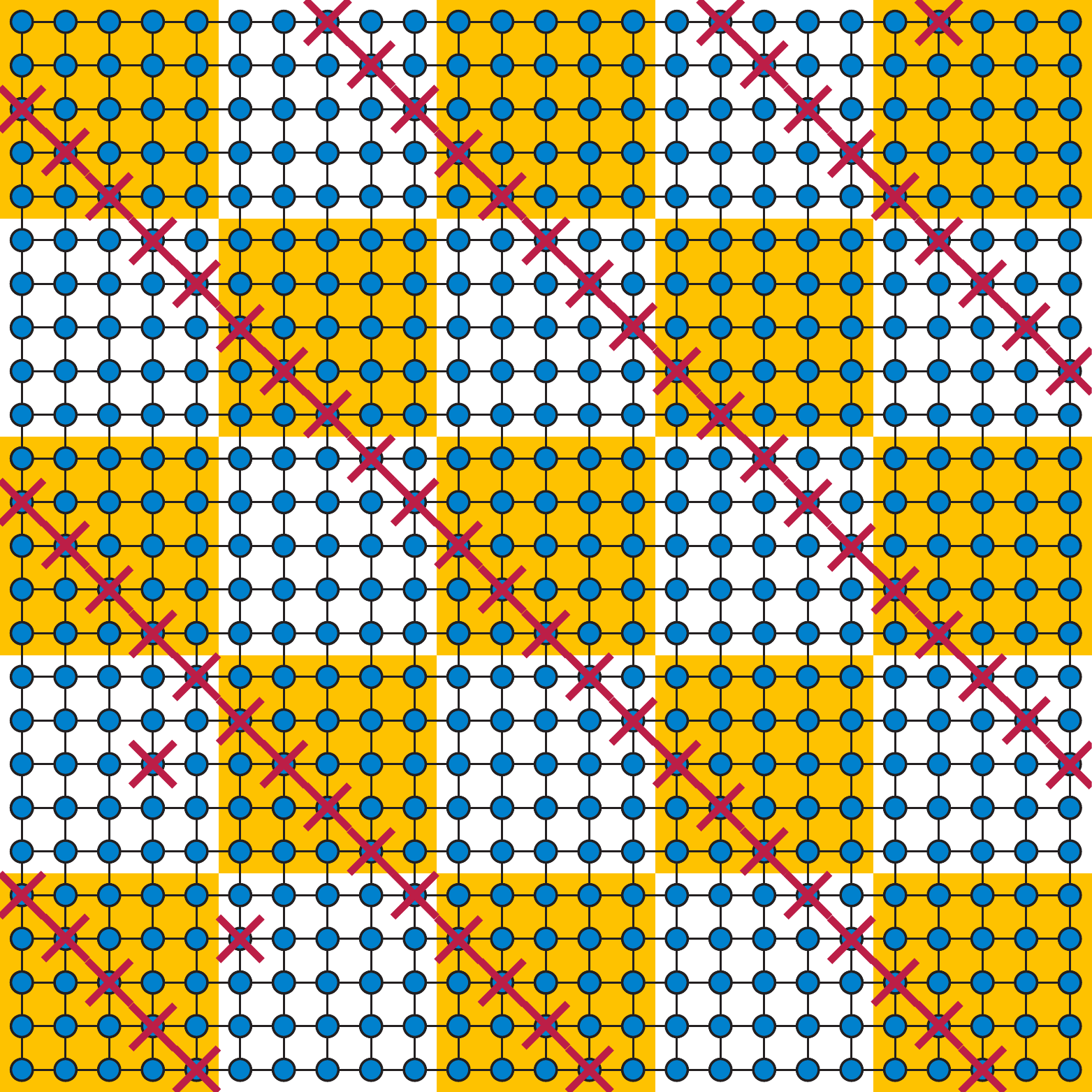}
\caption{Left: A $25\times 25$ grid with 69 damaged vertices in the pattern given by the proof of Theorem~\ref{thm:grid-minor-lb}. As shown in that theorem, its largest square grid minor is $4\times 4$. Right: Even with 72 damaged vertices, partitioning the grid into 25 $5\times 5$ subgrids leads to a subgrid with at most two damaged vertices, which necessarily has a $3\times 3$ grid minor.}
\label{fig:grid-subgrids}
\end{figure}

\begin{theorem}
\label{thm:grid-minor}
Suppose that an $n\times n$ grid is damaged by the deletion of a set $D$ of $m$ vertices.
Then the remaining graph has a grid minor of size $k\times k$, where $k=\max\{n-m,n^2/4m-O(1)\}=\Theta(\min\{n,n^2/m\})$.
\end{theorem}

\begin{proof}
To achieve $k=n-m$, observe that the rows and columns of the grid that are disjoint from $D$ form two sets of at least $n-m$ disjoint paths, as required by Lemma~\ref{lem:disjoint-paths}.

To achieve $k=n^2/4m-O(1)$, partition the given grid into approximately $4(m/n)^2$ subgrids of size approximately $n^2/(2m)\times n^2/(2m)$. The average number of damaged vertices per subgrid is $\tfrac{m}{4(m/n)^2}=n^2/4m$. There exists at least one subgrid whose number of damaged vertices is at most this average, which is half of the side length of the subgrid. Within this subgrid we may apply the $n'-m'$ bound (where $n'$ is the side length of the subgrid and $m'$ is its number of damaged vertices) giving us a minor of size approximately $n^2/(2m)-n^2/4m=n^2/4m$.

Finally, we observe that if $m<n/2$ then $\max\{n-m,n^2/4m-O(1)\}=\Theta(n)$ while if $m\ge n/2$ then $\max\{n-m,n^2/4m-O(1)\}=\Theta(n^2/m)$, so our bound achieves the stated asymptotics.
\end{proof}

\begin{corollary}
Let $\mathcal{F}$ be a fixed minor-closed family of graphs, let $G$ be a graph of treewidth $t$ in $\mathcal{F}$, and let $G'$ be formed by deleting $k$ vertices of $G$.
Then the treewidth of $G'$ is $\Omega(\min(t,t^2/k))$.
\end{corollary}

\begin{proof}
This follows immediately from Theorem~\ref{thm:grid-minor} and from the fact that, for graphs in $\mathcal{F}$, the treewidth is both upper bounded and lower bounded by linear functions of the grid minor size.
\end{proof}

As an example of Theorem~\ref{thm:grid-minor}, with a $25\times 25$ grid and $72$ damaged vertices, we may partition the grid into $25$ subgrids of size $5\times 5$. The average number of damaged vertices per subgrid is $72/25<3$, so there exists a subgrid with at most two damaged vertices, within which the undamaged rows and columns form sets of disjoint paths that give a $3\times 3$ grid minor; see Figure~\ref{fig:grid-subgrids}(right).

\begin{figure}[t]
\centering\includegraphics[width=\textwidth]{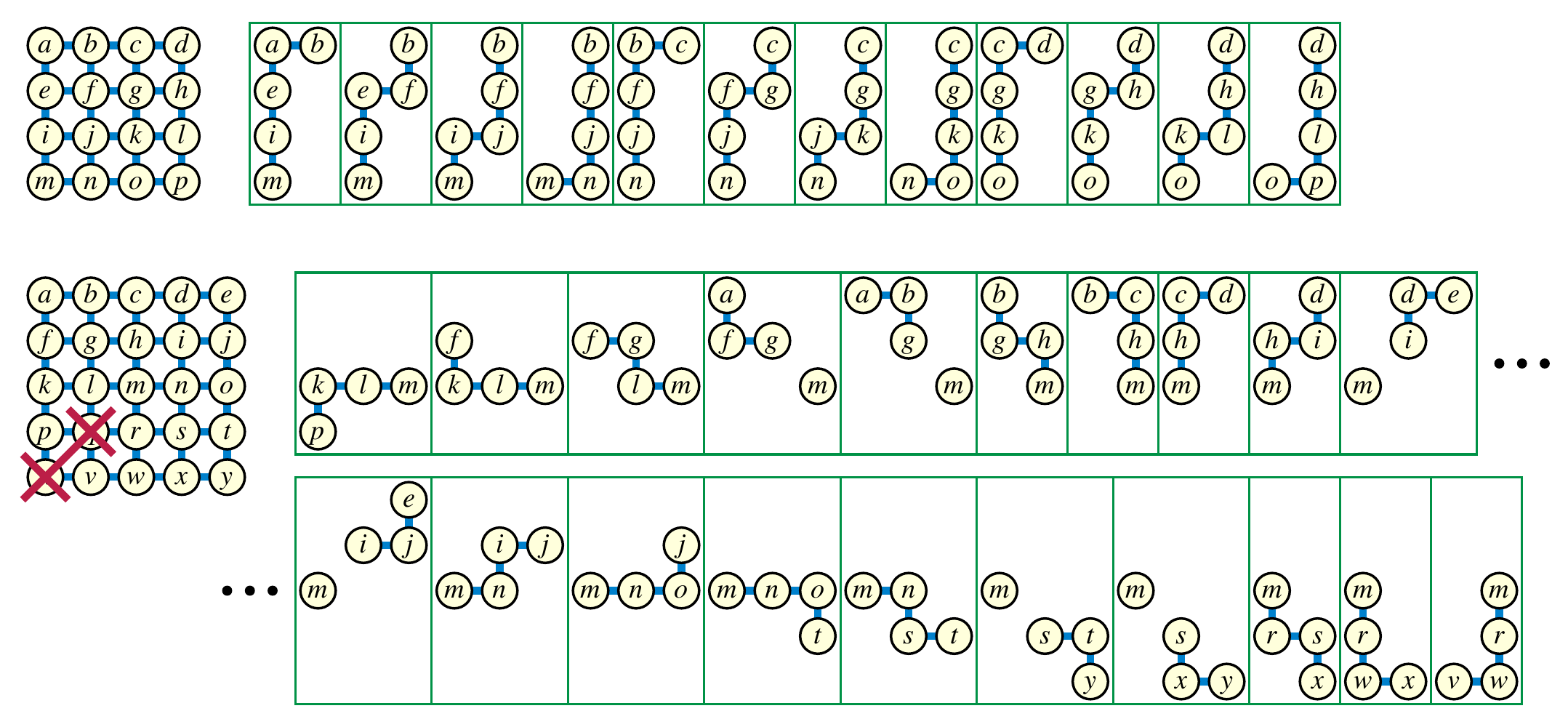}
\caption{Path decompositions. Top: decomposition of a $4\times 4$ grid with width~$4$. Bottom: decomposition of a $5\times 5$ grid with two damaged vertices, with width~$3$. Both of these decompositions have minimum width, equal to the pathwidth of their graphs.}
\label{fig:path-decompositions}
\end{figure}

An upper bound of the same form, $\Theta(\max\{n,n^2/m\})$, may be achieved by making the damaged set consist of a subset of the rows and columns of the grid, spaced far enough apart to make the total size of the damaged set be $m$. Every remaining connected component is itself a grid of the given size. However, we may achieve a better constant factor in this bound by applying the concepts of \emph{path decompositions} and \emph{pathwidth}~\cite{RobSey-JCTB-83}. A path decomposition of a graph is a sequence of subsets of vertices (called \emph{bags}) with the two properties that every vertex appears in a contiguous subsequence of bags and that the two endpoints of every edge appear in at least one bag. The width of a path decomposition is one less than the size of the largest bag, and the pathwidth of a graph is the minimum width of any of its path decompositions. Pathwidth is minor-monotonic (the pathwidth of a minor of $G$ is no more than the pathwidth of $G$), and the $n\times n$ grid has pathwidth exactly~$n$ (Figure~\ref{fig:path-decompositions}, top)~\cite{Bod-TCS-98}, so we can prevent a damaged grid from having large grid minors by making the induced subgraph of the undamaged vertices have small pathwidth.

Using pathwidth, we can show that our bound for grid minors in the case that $m$ is small is exact:

\begin{theorem}
For $m\le n/2$, there exists a set $D$ of $m$ damaged vertices in an $n\times n$ grid, such that the largest square grid minor using only undamaged vertices has size $(n-m)\times(n-m)$.
\end{theorem}

\begin{proof}
We choose $D$ to be a set of $m$ vertices extending diagonally from one corner of the grid towards the center, as shown in Figure~\ref{fig:path-decompositions}(bottom). We claim that the resulting damaged grid has pathwidth exactly $n-m$, and so cannot contain a square grid minor of larger than the stated size. In the case when $m=\lfloor n/2\rfloor$, a path decomposition with width $n-m$ may be found by putting the center vertex into all bags (in the case $n$ is odd) and otherwise making the ordering of the first bag containing each vertex be the radial sorted ordering of these vertices around the center (starting clockwise of the damaged diagonal and breaking ties arbitrarily), as shown in the figure. If $m$ is smaller than $\lfloor n/2\rfloor$, a decomposition with pathwidth $n-m$ may be obtained by using this same pattern for all vertices that do not lie on the diagonal line segment between the damaged corner and the center, and by including the undamaged vertices that do lie on this line segment into all bags.
\end{proof}

We also obtain an upper bound for larger values of $m$ that differs from our lower bound only by a factor of two:

\begin{theorem}
\label{thm:grid-minor-lb}
For $m>n/2$, there exists a set $D$ of $m$ damaged vertices in an $n\times n$ grid, such that the largest square grid minor using only undamaged vertices has size $k\times k$, with $k=\left\lceil\tfrac{n^2}{2m}-\tfrac12\right\rceil$.
\end{theorem}

\begin{proof}
Number the negatively-sloped diagonals of the grid from $1$ to $2n-1$, choose a number $r$, and damage all the vertices that belong to diagonals numbered $r$ modulo $2k+1$. At least one of the residue classes of diagonals modulo $2k+1$ must have at most $\tfrac{n^2}{2k+1}\le m$ damaged vertices in it. The pathwidth of the remaining sets of $2k$ contiguous diagonals is at most $k$, as may be shown by a path decomposition that sorts the vertices by the linear combination $x-y$ of their Cartesian coordinates, breaking ties arbitrarily. Therefore, the largest square grid minor in the remaining graph can have size at most $k\times k$.
\end{proof}

Figure~\ref{fig:grid-subgrids}(left) shows an example of this diagonal damage pattern for $n=25$ and $k=4$. The pattern of damage shown in the figure reduces the pathwidth of the remaining graph to~4, so its largest undamaged square grid minor has size $4\times 4$.
Theorem~\ref{thm:grid-minor-lb} shows that damaging at most 70 damaged vertices in a $25\times 25$ grid can block the existence of an undamaged $5\times 5$ grid minor, but as the figure shows a careful choice of $r$ leads to only 69 damaged vertices.
\section{Shallow minors}

The \emph{radius} of a graph $G$ is smallest number $r$ such that all vertices of $G$ are within distance $r$ of one of its vertices, called its \emph{center}. That is, the radius is
$$\min_{v\in V(G)}\max_{w\in V(G)} \operatorname{distance}_G(v,w).$$ If $V_i$ are disjoint sets of vertices, each of which induces a subgraph of radius at most $\lambda$, then a minor of $G$ formed by contracting each set $V_i$ to a single vertex, deleting vertices not in any set $V_i$, and possibly deleting some of the resulting edges, is called a \emph{shallow minor} of $G$, at depth $\lambda$. Thus, a shallow minor of unbounded depth (or of depth at least $n-1$) is just a minor, while a shallow minor of depth $0$ is exactly a subgraph. The results of the previous two sections, on grid subgraphs and grid minors, naturally raise the question of how large a square grid can be found as a shallow minor of a given depth $\lambda$ in a damaged grid.

\begin{theorem}
\label{thm:shallow}
Let $D$ be a set of $m$ vertices in an $n\times n$ grid, and let $\lambda\ge 1$ be given. Then there exists a shallow square grid minor at depth $\lambda$ in the grid, disjoint from $D$, of size $k\times k$, where $k=\Omega(\min\{n,n\sqrt{\tfrac{\lambda}{m}},\tfrac{n^2}{m}\})$. For every $n$, $m$, and $\lambda$, it is possible to choose $D$ in such a way that the largest shallow grid minor of the grid has side length $O(\min\{n,n\sqrt{\tfrac{\lambda}{m}},\tfrac{n^2}{m}\})$. The constants in the $O$- and $\Omega$-notation used here do not depend on $\lambda$, $m$, or $n$. 
\end{theorem}

\begin{figure}[t]
\centering\includegraphics[height=2in]{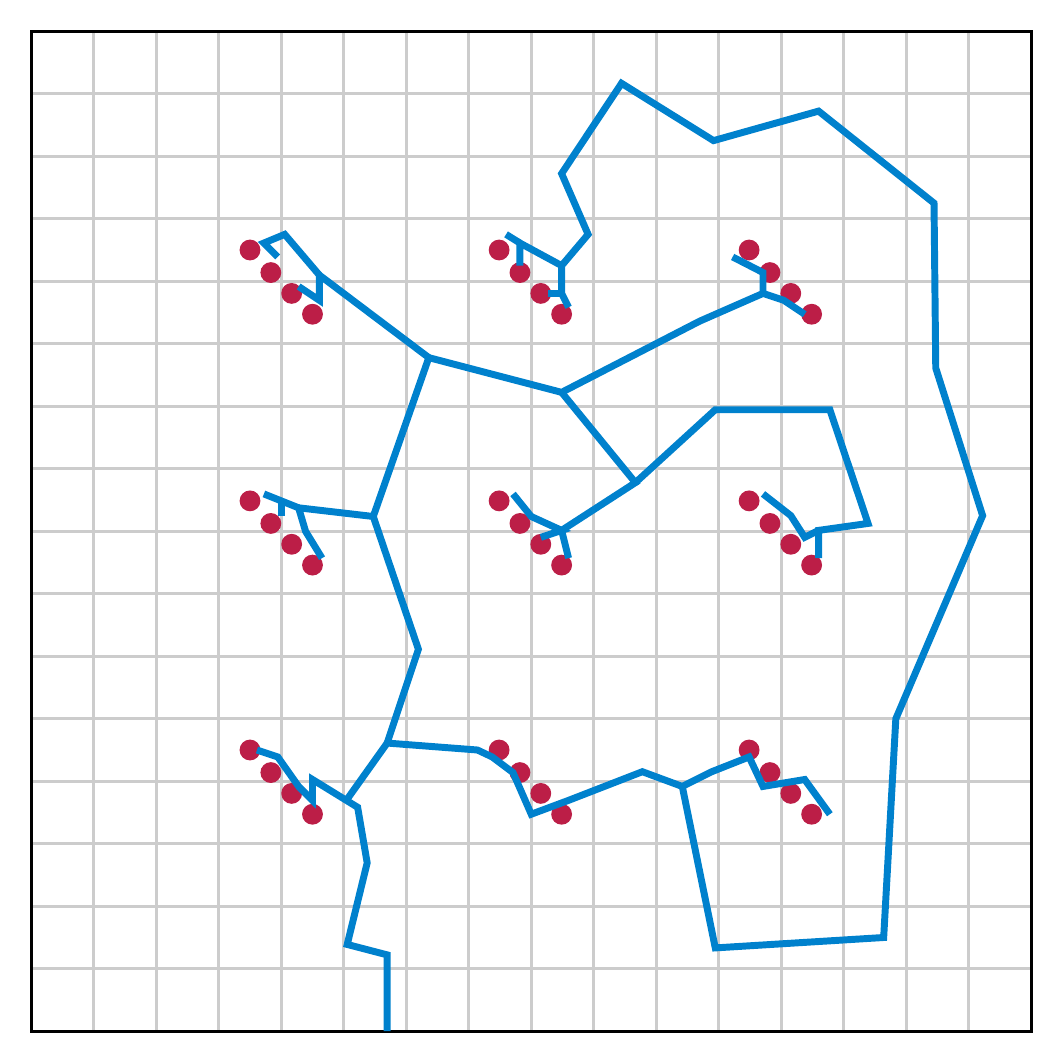}\qquad
\includegraphics[height=2in]{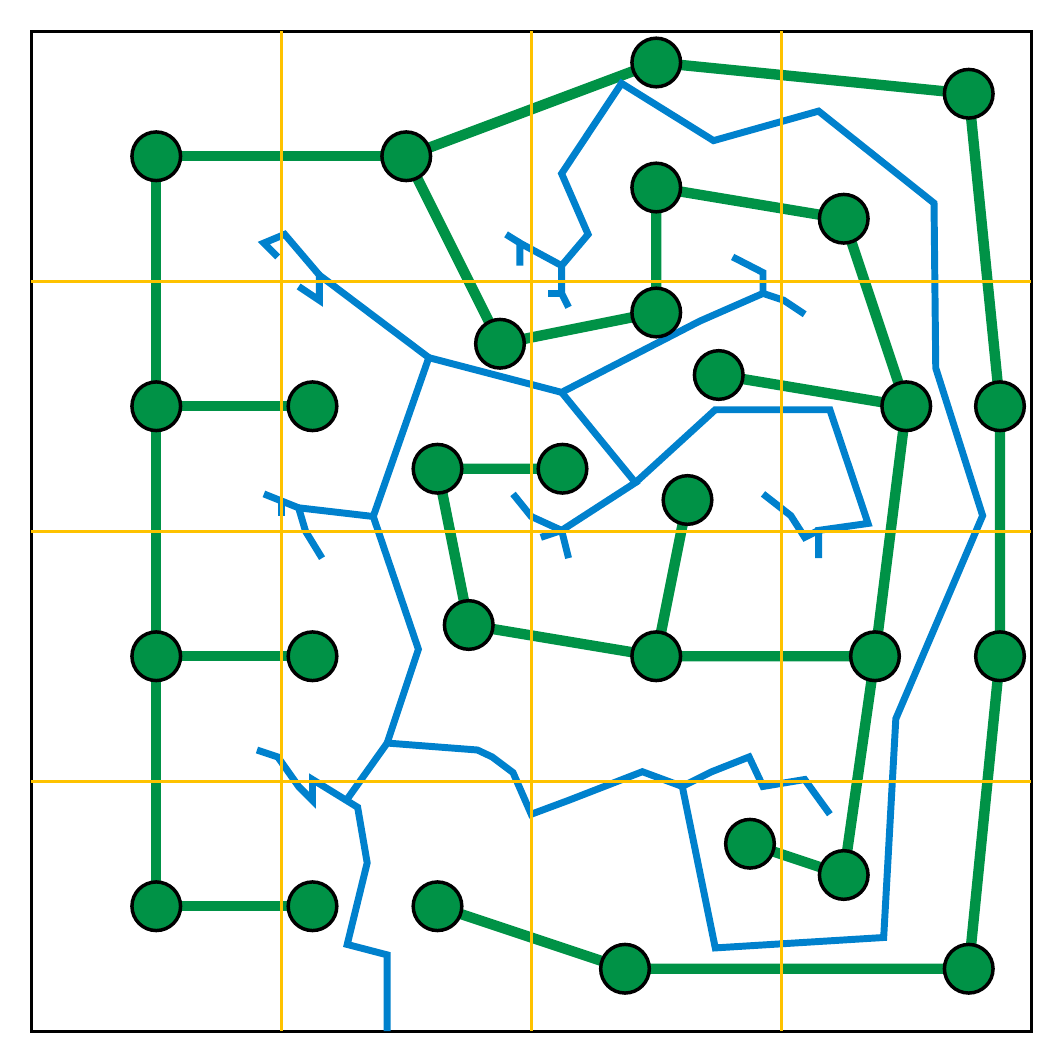}
\caption{Schematic view of upper bound construction for shallow minors. Left: The damaged vertices form a collection of diagonal segments of length $2\lambda+1$, evenly spaced throughout the grid (red). The tree $T_H$ in the dual graph, connecting faces near these segments and avoiding shallow minor $H$, is shown in blue. Right: Partitioning the grid into subgrids with corners at the centers of the damaged segments (yellow), and the tree $K$ representing the connected components of subgrids after $T_H$ is deleted (green).}
\label{fig:shallow-minor-lb}
\end{figure}

\begin{proof}
We assume $\lambda<n/2$, for otherwise every minor is a shallow minor and the result follows from Theorem~\ref{thm:grid-minor}.
To prove the existence of a large shallow grid minor we consider the following three cases.
\begin{itemize}
\item If $m\le\min(n/2,2\lambda)$, then the grid minor formed by the undamaged rows and columns of the grid has size $(m-n)\times (m-n)$, which is at least $n/2\times n/2$. The distance between two neighboring intersections of an undamaged row with an undamaged column is at most $2\lambda$, so the depth of the minor is at most $\lambda$.
\item If $m>\min(n/2,2\lambda)$ and $\lambda\le cn^2/m$ for a sufficiently small constant $c$, we may partition the grid into at least $m/2\lambda+1$ subgrids of side length $\Theta(n\sqrt{\lambda/m})$ (with a constant factor in the $\Theta$ notation that does not depend on $c$). One of these subgrids has at most $2\lambda$ damaged vertices in it, and its undamaged rows and columns form a shallow grid minor of depth at most $\lambda$ as in the first case. The assumption that $\lambda\le cn^2/m$ implies that $\lambda\le n\sqrt{c\lambda/m}$. Therefore, if $c$ is smaller than half the constant in the $\Theta$ notation, when we subtract the $2\lambda$ damaged rows and columns from the $\Theta(n\sqrt{\lambda/m})$ side length of the subgrid, the number of remaining undamaged rows and columns will still be $\Omega(n\sqrt{\lambda/m})$.
\item In the remaining case, $\lambda>cn^2/m$. We may apply the same method as in the proof of Theorem~\ref{thm:grid-minor}, which partitions the input grid into
smaller grids of side length $\Theta(n^2/m)$, one of which will have few enough damaged points that its remaining undamaged rows and columns will form a grid minor of side length $\Omega(n^2/m)$. By decreasing the size of these smaller grids by a constant factor (proportional to $1/c$), if necessary, we may ensure that their side length is at most $2\lambda$,
so that the resulting grid minor will automatically be shallow enough.
\end{itemize}

It remains to show that these bounds are tight, by exhibiting a set $D$ that blocks the existence of large shallow grid minors. This is trivial ($D$ can be empty) when $\min\{n,n\sqrt{{\lambda}/{m}},{n^2}/{m}\}=n$, and the existence of a suitable $D$ follows from Theorem~\ref{thm:grid-minor-lb} when the minimum is $n^2/m$. The remaining case, when the minimum is $n\sqrt{{\lambda}/{m}}$, combines ideas of planar embedding, treewidth, and interdigitating trees. Choose $D$ to be the vertices in at most $\lfloor m/(2\lambda+1)\rfloor$ diagonal line segments, each containing $2\lambda+1$ vertices, with these segments evenly spaced across the grid~$G$ as shown in Figure~\ref{fig:shallow-minor-lb}(left). Any minor of $G$ inherits a planar embedding from the embedding of $G$, but a square grid minor has a unique planar embedding up to the choice of its outer face, in which all faces other than the outer face are 4-cycles. None of these 4-cycles can completely surround one of the diagonal line segments in $D$ without exceeding the depth constraint, so all of the diagonal line segments of $D$ must lie within the outer face of any shallow grid minor. By similar reasoning the outer face of $G$ itself must also lie within the outer face of any shallow grid minor. It follows that we can find a connected tree $T_H$ in the dual graph of $G$, spanning the outer face of $G$ and all of the faces of $G$ that are adjacent to vertices of $D$, with the property that $T_H$ is disjoint from the edges of $G$ that have both endpoints in sets $V_i$. Each shallow grid minor $H$ of $G\setminus D$ is also a minor of the graph $G\setminus T_H$ formed by deleting the edges of $G$ that are dual to edges in $T_H$. The tree $T_H$ is illustrated in Figure~\ref{fig:shallow-minor-lb}(left).

No matter how such a tree $T_H$ is formed, the resulting graph $G\setminus T_H$ will have treewidth $O(n\sqrt{\lambda/m})$. For, consider the partition of $G$ into subgrids with corners at the centers of the diagonal damaged segments, and form a graph $K$ in which the vertices are connected components of the intersection of subgrids with $G\setminus T_H$ and the edges are adjacent pairs of components (Figure~\ref{fig:shallow-minor-lb}(right)). $K$ is a tree (it cannot contain a cycle, for such a cycle would surround one of the damaged diagonal segments preventing it from being attached to the rest of $T_H$), and hence has treewidth one. A tree-decomposition of $G\setminus T_H$ of width $O(n\sqrt{\lambda/m})$ may be obtained by using the decompositions of the subgrid components to expand the tree-decomposition of $K$.

We have seen that any shallow square grid minor $H$ of $G$ is a minor of $G\setminus T_H$, and that the treewidth of $G\setminus T_H$ and of its grid minors is $O(n\sqrt{\lambda/m})$. But the tree\-width of a grid is its side length~\cite{Bod-TCS-98}, so the largest shallow square grid minor in $G\setminus T_H$ must have side length $O(n\sqrt{\lambda/m})$.
\end{proof}

We remark that it does not seem possible to simplify the lower bound construction of this theorem by using pathwidth instead of treewidth, as we did for non-shallow minors. Indeed, for the set $D$ of vertices shown in Figure~\ref{fig:shallow-minor-lb}(left), there exist shallow minors of the remaining graph that do not surround any vertex of $D$ and whose pathwidth is $\Omega(n\log n\sqrt{\lambda/m})$, too high to give the bounds we need.

\section{Grids of bounded dimension}

Cubical grids of dimension higher than two (i.e. Cartesian products of equal-length paths) are less central than planar grids to the theory of graph minors, but there has been some past study of properties of these graphs related to minors~\cite{ChaKav-DM-06,ChaSiv-DM-07,OtaSud-DM-11}. As we show in this section, many of our results for planar grids extend directly to grids of higher dimension. Throughout this section, when we use $O$-notation, we treat the dimension $d$ as a fixed constant (suppressing the dependence on $D$ of the constant factors in the $O$-notation).

\begin{theorem}
Let $G$ be a $d$-dimensional cubical grid of side length $n$ with a set $D$ of $m$ damaged vertices. Then $G$ has a $d$-dimensional cubical grid subgraph of side length $\Omega(n/m^{1/d})$ disjoint from $D$. There exist sets $D$ with $|D|=m$ for which every $d$-dimensional cubical grid subgraph disjoint from $D$ has side length $O(n/m^{1/d})$.
\end{theorem}

\begin{proof}
As in Theorem~\ref{thm:subgraph}, the lower bound partitions the grid into more than $m$ subgrids of the given side length, at least one of which must be undamaged. The upper bound places the vertices of $D$ on points with Cartesian coordinates that are all $-1$ modulo $k+1$, for $k$ chosen to make $|D|\le m$.
\end{proof}

\begin{theorem}
\label{thm:bounded-dim-minor}
Let $G$ be a $d$-dimensional cubical grid of side length $n$ with a set $D$ of $m$ damaged vertices. Then $G$ has a $d$-dimensional cubical grid minor of side length $\Omega(\min\{n,(n^d/m)^{1/(d-1)}\})$ disjoint from $D$. There exist sets $D$ with $|D|=m$ for which every $d$-dimensional cubical grid minor disjoint from $D$ has side length $O(n^d/m)$.
\end{theorem}

\begin{proof}
For the lower bound, a grid with $m<n$ has a grid minor of side length $n-m$, formed by the intersection pattern of its undamaged $(d-1)$-dimensional axis-parallel hyperplanes. The result follows by partitioning the grid into subgrids whose average number of damaged vertices is proportional to their side length, and selecting a subgrid whose number of damaged vertices is at most average.

The upper bound places the vertices of $D$ on evenly spaced $(d-1)$-dimensional axis-parallel hyperplanes within the grid, partitioning it into subgrids of the given size that are disconnected from each other.
\end{proof}

Unlike the two-dimensional case, the upper and lower bounds of Theorem~\ref{thm:bounded-dim-minor} do not match, essentially because in higher dimensions the side length of a cube is no longer proportional to its surface measure.

\begin{theorem}
\label{thm:bounded-dim-shallow}
Let $G$ be a $d$-dimensional cubical grid of side length $n$ with a set $D$ of $m$ damaged vertices, and let $\lambda\ge 1$ be given. Then $G$ has a $d$-dimensional cubical grid shallow minor of depth $\lambda$ and side length
$$\Omega(\min\{n,n(\lambda/m)^{1/d},(n^d/m)^{1/(d-1)}\})$$
 disjoint from $D$.
\end{theorem}

\begin{proof}
Partition the grid into subgrids within one of which there must be at most 2$\lambda$ damaged vertices, and then find a minor using the remaining undamaged axis-parallel hyperplanes, as in Theorem~\ref{thm:shallow}.
\end{proof}

Theorem~\ref{thm:bounded-dim-shallow} generalizes the lower bound of Theorem~\ref{thm:shallow} to higher dimensions, but, we do not know how to generalize the corresponding upper bound.

\section{Hypercube subgraphs}

For grids of unbounded dimension, our knowledge is even more limited, but we can still prove some results. We define a \emph{hypercube graph} to be a $d$-dimensional cubical grid graph of side length $n=2$, and we let $N=2^d$ denote the number of vertices in such a graph. This graph may equivalently be described as the having as its vertex set the set of all length-$d$ binary strings, with edges that connect each two strings that differ in a single binary digit. In this setting it does not make sense to look for subgraphs or minors that are grid graphs of the same dimension but smaller size, as there is no nontrivial smaller size to look for; instead, we seek subgraphs or minors that are themselves hypercubes of large (but smaller) dimension. If the vertices of a $d$-dimensional hypercube are represented as length-$d$ sequences of the symbols $\{0,1\}$, then its hypercube subgraphs may similarly be represented as sequences of the symbols $\{0,1,\ast\}$, where the vertices in such a subgraph are obtained by replacing each $\ast$ symbol by either $0$ or $1$. The dimension of the hypercube subgraph is then its number of $\ast$ symbols.

There is a formal resemblance between the sets $D$ of vertices whose removal leaves no remaining large  hypercube  subgraph, and the sets $E$ of binary strings that form a good \emph{erasure code}~\cite{Riz-SIGCOMM-97,LubMitSho-ITIT-01,Lub-FOCS-02}. An erasure code capable of handling $e$ erasures can be described graph-theoretically as a subset $E$ of hypercube vertices such that every hypercube subgraph of dimension at most~$e$ contains at most one member of $E$. Instead, in our problem, there can be no undamaged hypercube of dimension~$d$ if the set $D$ of damaged vertices has the property that every hypercube subgraph of dimension at least~$d$ contains at least one member of $D$.

The smallest set $D$ whose removal eliminates all hypercube subgraphs of dimension $d-1$ from a $d$-dimensional hypercube has $|D|=2$: simply remove two opposite vertices. (In coding theory terms, this is a \emph{repetition code}.) However, to eliminate all hypercube subgraphs of dimension $d-2$, a set $D$ of non-constant size is needed.

\begin{lemma}
\label{lem:codim2-example}
For every even $m$, there exists a set $D$ of $m$ vertices within a hypercube of dimension
$d=\tbinom{m}{m/2}/2$ whose deletion eliminates all $(d-2)$-dimensional hypercube subgraphs.
\end{lemma}

\begin{proof}
There are exactly $d$ ways of partitioning the set $\{1,2,\dots,m\}$ into two subsets of equal size (with two partitions considered equivalent if they have the same two subsets in either order). Let the $i$th of these partitions be given by the two sets $D_i$ and $D\setminus D_i$, choosing $D_i$ arbitrarily among the two sets that define the partition. Let the $i$th coordinate of point $p_j$ in $D$ ($j=1,2,\dots m$) be $1$ if $j\in D_i$, and let the coordinate be $0$ otherwise.

Then every hypercube subgraph of dimension $d-2$ within the $d$-dimensional hypercube may be described as a string of $d$ symbols from $\{0,1,\ast\}$ of which exactly two symbols are not~$\ast$. Let $X$ and $Y$ be the two sets of $m/2$ points of $D$ whose coordinate values match these two symbols. Then, since $X$ and $Y$ are distinct and non-complementary sets of $m/2$ points, they have a non-empty intersection, a point of $D$ intersecting the given $(d-2)$-dimensional hypercube subgraphs. Since this hypercube subgraph was chosen arbitrarily, $D$ intersects all $(d-2)$-dimensional hypercube subgraphs.
\end{proof}

For example, let $D$ be a set of six points in a ten-dimensional hypercube, represented by the binary matrix
$$\left[\begin{array}{cccccccccc}
1&1&1&1&1&1&1&1&1&1\\
1&1&1&1&0&0&0&0&0&0\\
1&0&0&0&1&1&1&0&0&0\\
0&1&0&0&1&0&0&1&1&0\\
0&0&1&0&0&1&0&1&0&1\\
0&0&0&1&0&0&1&0&1&1\\
\end{array}\right]$$
where each column has equally many 0s and 1s, and all columns and their complements are distinct.
Then removing $D$ from the hypercube causes the largest remaining hypercube subgraph to have dimension seven. In general, the set $D$ constructed in the proof of the lemma may be interpreted as a binary code with size $m$ and length $\tbinom{m}{m/2}/2$, whose minimum distance
$\tbinom{m-2}{m/2-1}$ is smaller than its length by a factor of $\tfrac{m}{2m-2}>\tfrac{1}{2}$, matching the Plotkin bound on the ratio of distance and length for any code of size $m$~\cite{Plo-IRE-60}.

The bound of Lemma~\ref{lem:codim2-example} is tight:

\begin{lemma}
\label{lem:Sperner}
Let $D$ be a set of $m$ points in a $d$-dimensional hypercube, and suppose that $d>\tbinom{m}{\lfloor m/2\rfloor}/2$. Then there exists a $(d-2)$-dimensional hypercube subgraph disjoint from~$D$.
\end{lemma}

\begin{proof}
Let $\mathcal{H}$ be the family of $(d-1)$-dimensional hypercube subgraphs of the $d$-dimensional hypergraph, and let $\mathcal{F}$ be the multiset of $2d>\tbinom{m}{\lfloor m/2\rfloor}$ subsets of $D$ formed by intersecting each member of $\mathcal{H}$ with $D$.
By Sperner's theorem, there exist two sets $X$ and $Y$ in $\mathcal{F}$, corresponding to two distinct hypercube subgraphs $h_X$ and $h_Y$ in $\mathcal{F}$, such that $X\subseteq Y$.
Then the $(d-2)$-dimensional hypercube subgraph $h_x\cap \overline{h_y}$ (where $\overline h$ denotes the subgraph complementary to $h$) is disjoint from~$D$.
\end{proof}

\begin{theorem}
For every hypercube graph with $N$ vertices, and every set $D$ of $m$ damaged vertices, there is a hypercube subgraph disjoint from $D$ containing $\Omega(\min\{N,(N/m)\log\log(N/m)\})$ vertices.
\end{theorem}

\begin{proof}
If $m\le\log\log N$,
then $\tbinom{m}{\lfloor m/2\rfloor}<2^m\le\log n$ and we can apply Lemma~\ref{lem:Sperner} to find a hypercube subgraph disjoint from $D$ containing at least $N/4$ vertices.
Otherwise, 
choose a number $q\approx\log\log(N/m)$
such that $2^{2^q}<qN/2m$.
Let $p$ be the smallest power of two greater than $m/q$, partition the hypercube into $p$ disjoint smaller hypercubes, and choose one in which the number of damaged vertices is at most $q$.
Then there are at least $qN/2m$ vertices in this hypercube, allowing us to apply Lemma~\ref{lem:Sperner} within it to find a hypercube subgraph disjoint from $D$ containing $N/p=\Omega((N/m)\log\log(N/m))$ vertices.
\end{proof}

\begin{theorem}
\label{thm:hc-sg-ub}
For every hypercube graph with $N$ vertices, and every $m<D$, there exists a set $D$ with $|D|=m$ such that the number of vertices in the largest hypercube subgraph disjoint from $D$ is
$$O\left(\min\left\{
\frac{N\log N}{m},
\frac{N\log\tfrac{m}{\log\log N}\log\log N}{m}
\right\}\right).$$
\end{theorem}

\begin{proof}
We apply the probabilistic method to prove the existence of $D$. If $D$ is chosen uniformly at random, among all $m$-vertex sets, then the probability of a single vertex not belonging to $D$ is $1-m/N$ and the probability of a subcube of size $kN/m$ being disjoint from $D$ approximates $\exp(-k)$ from below. We will choose $k$ to be large enough that that this probability multiplied by the number of hypercube subgraphs of the appropriate size remains small.

Thus, we need to bound the number of hypercube subgraphs of the size given in the statement of the lemma. More strongly, we bound the number of hypercubes in a larger set, the hypercube subgraphs of size at least $N\log\log N/m$. There are at most $3^{\log N}$ hypercube subgraphs of the whole graph, obtained by specifying, for each coordinate, whether it is fixed to $0$ or $1$ or free to vary. If such a hypercube subgraph is to have size at least $N\log\log N/m$, then at most $\log(m/\log\log N)$ of the coordinates may be fixed, so alternatively, there are at most $(3\log N)^{\lceil\log(m/\log\log N)\rceil}$ such subgraphs, obtained by specifying the indices of exactly $\lceil\log\tfrac{m}{\log\log N}\rceil$ coordinates and, for each of them, whether that coordinate is fixed to $0$ or $1$ or free to vary. Thus, the number of hypercube subgraphs of  size at least $N/m$ is at most $\min\left\{3^{\log N},(3\log N)^{\lceil\log(m/\log\log N)\rceil}\right\}$.

Now choose $k$ to be larger than
$$\ln\min\left\{3^{\log N},(3\log N)^{\left\lceil\log\tfrac{m}{\log\log N}\right\rceil}\right\}
=
O\left(\min\left\{\log\log N,\log\tfrac{m}{\log\log N}\log\log N\right\}\right).$$
For such $k$, $\exp(-k)$ is smaller than the inverse of the number of hypercubes of size $kN/m$. Therefore, the expected number of subcubes of size $kN/m$ that are disjoint from $D$ is strictly less than one, but this number is the expected value of a non-negative integer random variable (the number of subcubes disjoint from $D$ for a particular random choice of $D$) so there must exist a choice of $D$ that makes the number of disjoint subcubes zero.
\end{proof}

\section{Hypercube minors}

Next, we consider hypercube minors instead of hypercube subgraphs.
We say that coordinate $i$ is a \emph{bad coordinate} for a set $D$ of damaged vertices of a hypercube if there are two vertices in $D$ at distance at most two from each other, such that the one or two coordinates on which these vertices differ include~$i$. For instance, for the set $D$ of six vertices in a ten-dimensional hypercube given earlier, there are no bad coordinates: all members of $D$ are at distance exactly six from each other. In a $d$-dimensional hypercube for which $i$ is not a bad coordinate, contracting all edges between pairs of vertices that differ in coordinate $i$ results in an undamaged hypercube minor of dimension $d-1$. This is because every contracted vertex is formed from at least one undamaged vertex of the original hypercube, and every edge of the contracted hypercube corresponds to at least one edge between two undamaged vertices of the original hypercube.

\begin{lemma}
\label{lem:num-bad-coords}
For every hypercube and every set $D$ of $m$ vertices, there are at most $2m-2$ bad coordinates.
\end{lemma}

\begin{proof}
Draw a graph $F$ on $D$, connecting a subset of the pairs of vertices at distance two from each other, with the subset chosen to be minimal with the property that every bad coordinate~$i$ is represented by an edge in $B$ between two vertices that differ in coordinate~$i$. Then $F$ must be acyclic, for if an edge of $e$ belonged to a cycle then the bad coordinates represented by~$e$ would also each be represented by at least one other edge of the cycle. Therefore, $F$ is a forest, with at most $m-1$ edges. Each bad coordinate is covered by an edge in $F$, and every edge covers at most two bad coordinates, so the number of bad coordinates is at most $2m-2$.
\end{proof}

An example showing Lemma~\ref{lem:num-bad-coords} to be tight may be constructed by letting $D$ consist of the origin and of $m-1$ points at distance two from it, no two sharing the same nonzero coordinate.

\begin{theorem}
\label{thm:hc-minor}
For every hypercube graph with $N$ vertices, and every set $D$ of $m$ damaged vertices, there is a hypercube shallow minor of depth 1 disjoint from~$D$ and containing $\Omega(\min\{N,(N/m)\log(N/m)\})$ vertices.
\end{theorem}

\begin{proof}
If $m\le \tfrac12\log N$, by Lemma~\ref{lem:num-bad-coords} there is a coordinate that is not bad, and contracting that coordinate produces a hypercube minor with $N/2$ vertices. Otherwise, let $p$ be the smallest power of two greater than $2m/\log(N/2m)$, partition into $p$ disjoint smaller hypercubes, and choose one in which the number of damaged vertices is at most $\tfrac12\log(N/2m)$. The number of vertices in this hypercube is large enough that there necessarily exists a coordinate that is not bad, and again contracting that coordinate produces the desired minor.
\end{proof}

\begin{figure}[t]
\centering\includegraphics[width=2in]{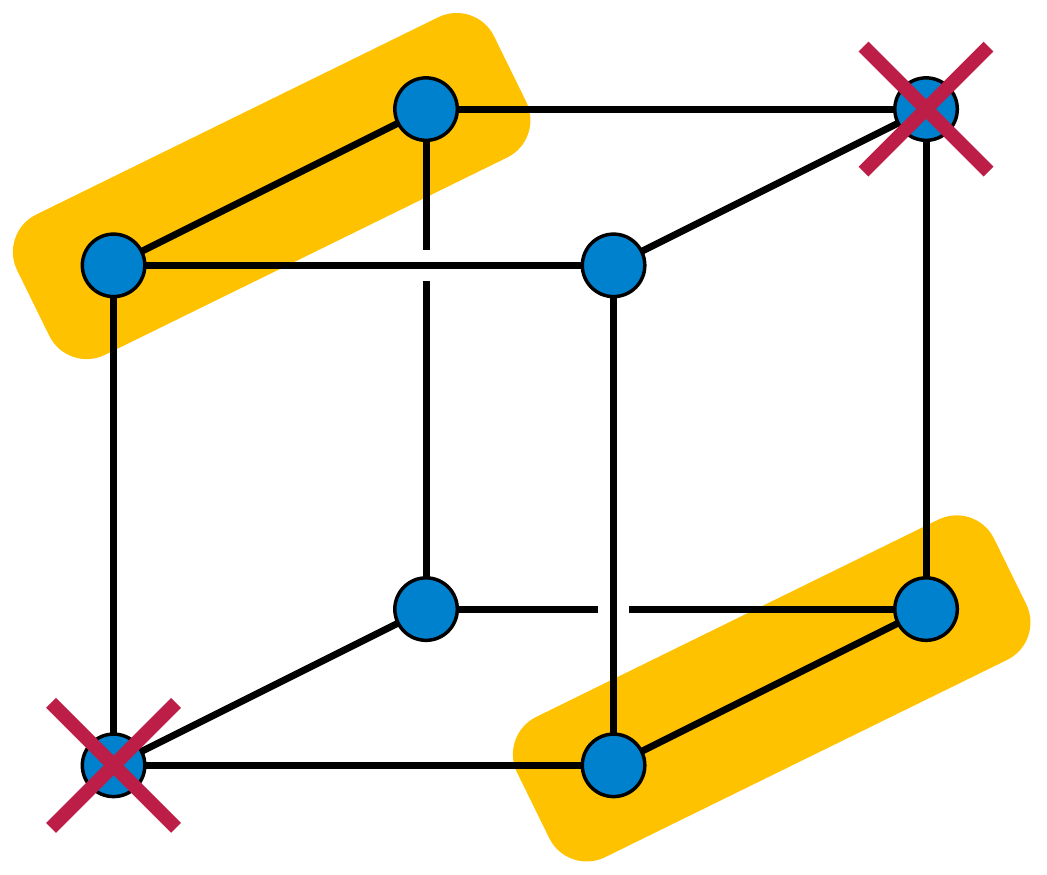}
\caption{Deleting two opposite vertices of a hypercube (here of dimension $d=3$, with the deleted vertices marked by red crosses) eliminates all ($d-1$)-dimensional hypercube subgraphs, but contracting the edges (shaded in yellow) that correspond to a non-bad coordinate produces a ($d-1$)-dimensional hypercube minor.}
\label{fig:hc-sg-vs-m}
\end{figure}
 
For example, with $N=8$ and $m=2$, choosing $D$ to be the two opposite corners of a cube reduces the number of vertices in the largest hypercube subgraph to be only two. However, no matter how $D$ is chosen, at most two of the three coordinates may be bad, so there is always a way to contract one coordinate and produce a square minor with four vertices (Figure~\ref{fig:hc-sg-vs-m}).

Unfortunately we have no nontrivial upper bound on the size of hypercube minors, corresponding to the bound of Theorem~\ref{thm:hc-sg-ub} on hypercube subgraphs. However, although our upper and lower bounds for hypercube subgraphs and minors are not far from each other, they are  strong enough to show that the functions mapping $N$ and $m$ to the size of the largest hypercube subgraph and the size of the largest hypercube minor respectively may differ by more than a constant factor. For instance, for $m=\Theta(\log N)$, Theorem~\ref{thm:hc-minor} shows that there is still a hypercube minor of size $\Omega(N)$, while Theorem~\ref{thm:hc-sg-ub} shows that the largest hypercube subgraph may be forced to be smaller by a factor of at least
$$\Theta\left(
\frac{\log N}{(\log\log N)^2}\right).$$

\subsection*{Acknowledgements}
This research was supported in part  by the National Science Foundation under grants 0830403 and 1217322, and by the Office of Naval Research under MURI grant N00014-08-1-1015. I thank Tasos Sidiropoulos for suggesting this problem, and both Tasos and Chandra Chekuri for encouraging me to write up these results as a paper.

{\raggedright
\bibliographystyle{abuser}
\bibliography{grid-minors}}

\begin{thebibliography}{10}

\bibitem{BirBonRee-DAM-09}
E.~Birmel{\'e}, J.~A. Bondy, and B.~A. Reed.
\newblock {Tree-width of graphs without a $3\times3$ grid minor}.
\newblock {\em Discrete Appl. Math.} 157(12):2577{--}2596, 2009,
  \href{http://dx.doi.org/10.1016/j.dam.2008.08.003}%
{doi:10.1016/j.dam.2008.08.003}.

\bibitem{Bod-TCS-98}
H.~L. Bodlaender.
\newblock {A partial $k$-arboretum of graphs with bounded treewidth}.
\newblock {\em Theoret. Comput. Sci.} 209(1-2):1{--}45, 1998,
  \href{http://dx.doi.org/10.1016/S0304-3975(97)00228-4}%
{doi:10.1016/S0304-3975(97)00228-4}.

\bibitem{ChaKav-DM-06}
L.~S. Chandran and T.~Kavitha.
\newblock {The treewidth and pathwidth of hypercubes}.
\newblock {\em Discrete Math.} 306(3):359{--}365, 2006,
  \href{http://dx.doi.org/10.1016/j.disc.2005.12.010}%
{doi:10.1016/j.disc.2005.12.010}.

\bibitem{ChaSiv-DM-07}
L.~S. Chandran and N.~Sivadasan.
\newblock {On the Hadwiger's conjecture for graph products}.
\newblock {\em Discrete Math.} 307(2):266{--}273, 2007,
  \href{http://dx.doi.org/10.1016/j.disc.2006.06.019}%
{doi:10.1016/j.disc.2006.06.019}.

\bibitem{CheChu-13}
C.~Chekuri and J.~Chuzhoy.
\newblock {Polynomial bounds for the grid-minor theorem}.
\newblock Electronic preprint arxiv:1305.6577, 2013.

\bibitem{CheSid-ms-12}
C.~Chekuri and A.~Sidiropoulos.
\newblock {Approximation algorithms for Euler genus and related problems}.
\newblock Electronic preprint arxiv:1304.2416, 2013.
\newblock To appear at 54th IEEE Symp. Foundations of Computer Science, 2013.

\bibitem{ColMagSit-SJC-97}
R.~J. Cole, B.~M. Maggs, and R.~K. Sitaraman.
\newblock {Reconfiguring arrays with faults, part I: worst-case faults}.
\newblock {\em SIAM J. Comput.} 26(6):1581{--}1611, 1997,
  \href{http://dx.doi.org/10.1137/S0097539793255011}%
{doi:10.1137/S0097539793255011}.

\bibitem{DemHaj-Comb-08}
E.~D. Demaine and M.~Hajiaghayi.
\newblock {Linearity of grid minors in treewidth with applications through
  bidimensionality}.
\newblock {\em Combinatorica} 28(1):19{--}36, 2008,
  \href{http://dx.doi.org/10.1007/s00493-008-2140-4}%
{doi:10.1007/s00493-008-2140-4}.

\bibitem{DemHajKaw-Algo-09}
E.~D. Demaine, M.~Hajiaghayi, and K.~Kawarabayashi.
\newblock {Algorithmic graph minor theory: improved grid minor bounds and
  Wagner's contraction}.
\newblock {\em Algorithmica} 54(2):142{--}180, 2009,
  \href{http://dx.doi.org/10.1007/s00453-007-9138-y}%
{doi:10.1007/s00453-007-9138-y}.

\bibitem{Die-AMSUH-04}
R.~Diestel.
\newblock {A short proof of Halin's grid theorem}.
\newblock {\em Abh. Math. Sem. Univ. Hamburg} 74:237{--}242, 2004,
  \href{http://dx.doi.org/10.1007/BF02941538}%
{doi:10.1007/BF02941538}.

\bibitem{DieJenGor-JCTB-99}
R.~Diestel, T.~R. Jensen, K.~Y. Gorbunov, and C.~Thomassen.
\newblock {Highly connected sets and the excluded grid theorem}.
\newblock {\em J. Combin. Theory Ser. B} 75(1):61{--}73, 1999,
  \href{http://dx.doi.org/10.1006/jctb.1998.1862}%
{doi:10.1006/jctb.1998.1862}.

\bibitem{Gri-DMTCS-11}
A.~Grigoriev.
\newblock {Tree-width and large grid minors in planar graphs}.
\newblock {\em Discrete Math. Theor. Comput. Sci.} 13(1):13{--}20, 2011.

\bibitem{GuTam-Algo-12}
Q.-P. Gu and H.~Tamaki.
\newblock {Improved bounds on the planar branchwidth with respect to the
  largest grid minor size}.
\newblock {\em Algorithmica} 64(3):416{--}453, 2012,
  \href{http://dx.doi.org/10.1007/s00453-012-9627-5}%
{doi:10.1007/s00453-012-9627-5}.

\bibitem{Hal-MN-65}
R.~Halin.
\newblock {{\"U}ber die Maximalzahl fremder unendlicher Wege in Graphen}.
\newblock {\em Math. Nachr.} 30:63{--}85, 1965,
  \href{http://dx.doi.org/10.1002/mana.19650300106}%
{doi:10.1002/mana.19650300106}.

\bibitem{KakKarLei-FOCS-90}
C.~Kaklamanis, A.~R. Karlin, F.~T. Leighton, V.~Milenkovic, P.~Raghavan,
  S.~Rao, C.~Thomborson, and A.~Tsantilas.
\newblock {Asymptotically tight bounds for computing with faulty arrays of
  processors}.
\newblock {\em Proc. 31st IEEE Symp. Foundations of Computer Science (FOCS)},
  pp.~285{--}296, 1990, \href{http://dx.doi.org/10.1109/FSCS.1990.89547}%
{doi:10.1109/FSCS.1990.89547}.

\bibitem{KawKob-STACS-12}
K.~Kawarabayashi and Y.~Kobayashi.
\newblock {Linear min-max relation between the treewidth of $H$-minor-free
  graphs and its largest grid}.
\newblock {\em Proc. 29th International Symposium on Theoretical Aspects of
  Computer Science (STACS 2012)}, pp.~278{--}289, Leibniz Int. Proc.
  Informatics~14, 2012, \href{http://dx.doi.org/10.4230/LIPIcs.STACS.2012.278}%
{doi:10.4230/LIPIcs.STACS.2012.278}.

\bibitem{Lub-FOCS-02}
M.~G. Luby.
\newblock {LT codes}.
\newblock {\em Proc. 43rd IEEE Symp. Foundations of Computer Science},
  pp.~271{--}280, 2002, \href{http://dx.doi.org/10.1109/SFCS.2002.1181950}%
{doi:10.1109/SFCS.2002.1181950}.

\bibitem{LubMitSho-ITIT-01}
M.~G. Luby, M.~Mitzenmacher, M.~A. Shokrollahi, and D.~A. Spielman.
\newblock {Efficient erasure correcting codes}.
\newblock {\em IEEE Trans. Information Theory} 47(2):569{--}584, 2001,
  \href{http://dx.doi.org/10.1109/18.910575}%
{doi:10.1109/18.910575}.

\bibitem{NesOss-12}
J.~Ne{\v{s}}et{\v{r}}il and P.~Ossona~de Mendez.
\newblock {\em {Sparsity: Graphs, Structures, and Algorithms}}.
\newblock Algorithms and Combinatorics~28. Springer, 2012,
  \href{http://dx.doi.org/10.1007/978-3-642-27875-4}%
{doi:10.1007/978-3-642-27875-4}.

\bibitem{OtaSud-DM-11}
Y.~Otachi and R.~Suda.
\newblock {Bandwidth and pathwidth of three-dimensional grids}.
\newblock {\em Discrete Math.} 311(10-11):881{--}887, 2011,
  \href{http://dx.doi.org/10.1016/j.disc.2011.02.019}%
{doi:10.1016/j.disc.2011.02.019}.

\bibitem{Plo-IRE-60}
M.~Plotkin.
\newblock {Binary codes with specified minimum distance}.
\newblock {\em IRE Trans. Information Theory} 6:445{--}450, 1960,
  \href{http://dx.doi.org/10.1109/TIT.1960.1057584}%
{doi:10.1109/TIT.1960.1057584}.

\bibitem{PloRaoSmi-SODA-94}
S.~Plotkin, S.~Rao, and W.~D. Smith.
\newblock {Shallow excluded minors and improved graph decompositions}.
\newblock {\em Proc. 5th ACM-SIAM Symp. on Discrete Algorithms (SODA)},
  pp.~462{--}470, 1994.

\bibitem{Ree-SC-97}
B.~A. Reed.
\newblock {Tree width and tangles: a new connectivity measure and some
  applications}.
\newblock {\em Surveys in Combinatorics, 1997 (London)}, pp.~87{--}162.
  Cambridge Univ. Press, London Math. Soc. Lecture Note Ser. 241, 1997,
  \href{http://dx.doi.org/10.1017/CBO9780511662119.006}%
{doi:10.1017/CBO9780511662119.006}.

\bibitem{ReeWoo-EJC-12}
B.~A. Reed and D.~R. Wood.
\newblock {Polynomial treewidth forces a large grid-like-minor}.
\newblock {\em European J. Combin.} 33(3):374{--}379, 2012,
  \href{http://dx.doi.org/10.1016/j.ejc.2011.09.004}%
{doi:10.1016/j.ejc.2011.09.004}.

\bibitem{Riz-SIGCOMM-97}
L.~Rizzo.
\newblock {Effective erasure codes for reliable computer communication
  protocols}.
\newblock {\em SIGCOMM Comput. Commun. Rev.} 27(2):24{--}36, 1997,
  \href{http://dx.doi.org/10.1145/263876.263881}%
{doi:10.1145/263876.263881}.

\bibitem{RobSeyTho-JCTB-94}
N.~Robertson, P.~Seymour, and R.~Thomas.
\newblock {Quickly excluding a planar graph}.
\newblock {\em J. Combin. Theory Ser. B} 62(2):323{--}348, 1994,
  \href{http://dx.doi.org/10.1006/jctb.1994.1073}%
{doi:10.1006/jctb.1994.1073}.

\bibitem{RobSey-JCTB-83}
N.~Robertson and P.~D. Seymour.
\newblock {Graph minors. I. Excluding a forest}.
\newblock {\em J. Combin. Theory Ser. B} 35(1):39{--}61, 1983,
  \href{http://dx.doi.org/10.1016/0095-8956(83)90079-5}%
{doi:10.1016/0095-8956(83)90079-5}.

\bibitem{RobSey-JCTB-91}
N.~Robertson and P.~D. Seymour.
\newblock {Graph minors. X. Obstructions to tree-decomposition}.
\newblock {\em J. Combin. Theory Ser. B} 52(2):153{--}190, 1991,
  \href{http://dx.doi.org/10.1016/0095-8956(91)90061-N}%
{doi:10.1016/0095-8956(91)90061-N}.

\bibitem{Ten-CGTA-98}
S.-H. Teng.
\newblock {Combinatorial aspects of geometric graphs}.
\newblock {\em Comput. Geom.} 9(4):277{--}287, 1998,
  \href{http://dx.doi.org/10.1016/S0925-7721(96)00008-9}%
{doi:10.1016/S0925-7721(96)00008-9}.

\bibitem{Wul-FOCS-11}
C.~Wulff-Nilsen.
\newblock {Separator theorems for minor-free and shallow minor-free graphs with
  applications}.
\newblock {\em Proc. 52nd IEEE Symp. Foundations of Computer Science (FOCS)},
  pp.~37{--}46, 2011, \href{http://dx.doi.org/10.1109/FOCS.2011.15}%
{doi:10.1109/FOCS.2011.15}.

\end{thebibliography}

\end{document}